\documentclass[11pt]{article}

\usepackage[T1]{fontenc}
\usepackage[utf8]{inputenc}
\usepackage{lmodern}
\usepackage[english]{babel}
\usepackage[a4paper,margin=2.7cm]{geometry}
\usepackage{amsmath,amssymb,amsthm,mathtools}
\usepackage{booktabs,array}
\usepackage{float}
\usepackage{needspace}
\usepackage{microtype}
\usepackage{xcolor}
\usepackage{tikz}
\usepackage[round,authoryear]{natbib}
\usepackage{authblk}
\usepackage[colorlinks=true,linkcolor=blue!55!black,
  citecolor=blue!55!black,urlcolor=blue!55!black]{hyperref}
\hypersetup{
  pdftitle={Tail representations for absolute, signed-power, inverse, and logarithmic moments},
  pdfauthor={Roberto Vila and Eduardo Nakano},
  pdfsubject={Tail representations for moment functionals},
  pdfkeywords={tail integral, layer-cake representation, absolute moment,
    signed power, integer-valued random variable, inverse moment,
    logarithmic moment}
}

\newtheorem{theorem}{Theorem}[section]
\newtheorem{proposition}[theorem]{Proposition}
\newtheorem{corollary}[theorem]{Corollary}
\theoremstyle{definition}
\newtheorem{definition}[theorem]{Definition}
\newtheorem{example}[theorem]{Example}
\theoremstyle{remark}
\newtheorem{remark}[theorem]{Remark}

\newcommand{\E}{\mathbb{E}}
\newcommand{\Pp}{\mathbb{P}}
\newcommand{\R}{\mathbb{R}}

\newcommand{\ind}{\mathbf{1}}
\newcommand{\sgn}{\operatorname{sgn}}
\newcommand{\dd}{\,\mathrm{d}}
\newcommand{\abs}[1]{\left\lvert #1\right\rvert}
\newcommand{\spow}[2]{#1^{\langle #2\rangle}}
\newcommand{\doi}[1]{\href{https://doi.org/#1}{\nolinkurl{doi:#1}}}

\title{Tail representations for absolute, signed-power,\\
  inverse, and logarithmic moments}
\author[1,*]{Roberto Vila}
\author[1]{Eduardo Nakano}
\author[2]{Cecilia Castro}
\affil[1]{Department of Statistics, University of Bras\'ilia, Bras\'ilia, Brazil}
\affil[2]{Centre of Mathematics, Universidade do Minho, Braga, Portugal}
\date{August 2026}

\begin{document}
\maketitle
\begingroup
\renewcommand{\thefootnote}{\fnsymbol{footnote}}
\footnotetext[1]{Corresponding author:
  \href{mailto:rovig161@gmail.com}{rovig161@gmail.com}.}
\endgroup

\begin{abstract}
Let $X$ be a real-valued random variable and let $p>0$.  When $X$ can take
negative values, $X^p$ is not generally real-valued for non-integer $p$;
formulas for moments of arbitrary positive order must therefore distinguish
the absolute power
$\abs{X}^p$ from the signed power $\spow{X}{p}=\sgn(X)\abs{X}^p$.
Using the layer-cake identity, we give a self-contained treatment of tail
representations for the positive and negative parts of $X$, including
absolute and signed-power moments, centered forms, and exact formulas on
ordered discrete supports and integer lattices.  The basic identities require
neither a density nor continuity and apply to continuous, discrete, singular,
and mixed distributions.  For strictly positive random variables, we also
give representations and finiteness criteria for inverse moments, together
with distribution-function and Frullani--Laplace formulas for logarithmic
moments.  Examples illustrate two-sided support, atoms, lattice tails, and
behavior near zero.  The presentation places these classical identities in a
common notation and states the domain and integrability conditions required
in each case.
\end{abstract}


\section{Introduction}

For a nonnegative random variable $Y$ and $p>0$, the identity
\begin{equation}\label{eq:classical-tail}
  \E[Y^p]
  =p\int_0^\infty t^{p-1}\Pp(Y>t)\dd t
\end{equation}
is one of the most useful links between a moment and a distributional tail.
It is often called the tail-integral, integrated-tail, or layer-cake
representation.  It requires neither a density nor continuity.  It is also
valid as an identity in $[0,\infty]$, so it can be used to decide whether a
moment exists.

The tail-integral formula and its extensions have a substantial literature.
Classical probability texts record versions of \eqref{eq:classical-tail};
see, for example, \citet{Feller1971} and \citet{Shorack2000}.
\citet{Hong2012} discusses higher-order versions for nonnegative variables,
and \citet{Hong2015} treats product moments.  Continuous and integer-valued
formulas for nonnegative variables and orders $r\geq1$ are given by
\citet{Chakraborti2019},
while \citet[Proposition~2.1]{Lo2019} gives the two-sided formula for the
mean of an integrable real-valued variable and develops its geometric and
pedagogical interpretation.  Related tail-integral formulas for real-valued
random variables are studied by \citet{NadarajahOkorie2022}.  Broader
identities for transformed expectations and random vectors are developed by
\citet{Liu2020} and \citet{Ogasawara2020}, respectively.

Against this background, the present note has an expository objective: it
organizes several tail representations in a common notation, with explicit
attention to their domains, integrability conditions, and behavior at atoms.
Four distinctions are central to this organization.

First, a general real power of a negative number is not real.  The phrase
``real-order moment of a real-valued random variable'' is therefore
ambiguous unless the power convention is stated.  Second, a signed-power
moment cannot be obtained by subtracting two infinite one-sided moments.
Third, atoms require care with strict inequalities and left limits of the
cumulative distribution function.  Fourth, inverse moments are controlled
by behavior near zero, whereas positive-order moments are controlled by
the tails at infinity.

Accordingly, we work with positive and negative parts, absolute powers, and
signed powers.  The one-sided identities hold in $[0,\infty]$ for every
$p>0$ and combine into the extended-valued absolute-moment formula.  When
$\E[\abs{X}^p]<\infty$, they also give the signed-power moment.  Ordinary integer
moments are recovered as special cases.  Section~\ref{sec:general} develops
these representations and
their centered forms.  Section~\ref{sec:discrete} gives exact sums on ordered
discrete supports and integer lattices.  Section~\ref{sec:inverse} treats
inverse moments of strictly positive random variables.
Section~\ref{sec:logarithmic}
gives distribution-function and Laplace-transform representations of
$\E[\log X]$.  Section~\ref{sec:summary} summarizes the relevant domains and
finiteness conditions.

\section{Tail representations for absolute and signed-power moments}
\label{sec:general}

Let $(\Omega,\mathcal F,\Pp)$ be a probability space and let
$X:\Omega\to\R$ be a random variable.  Write
\[
  X_+=\max(X,0),
  \qquad
  X_-=\max(-X,0),
\]
so that $X=X_+-X_-$ and $\abs{X}=X_++X_-$.  Let
\[
  F_X(x)=\Pp(X\leq x),
  \qquad
  F_X(x^-)=\Pp(X<x),
  \qquad
  \overline F_X(x)=1-F_X(x)=\Pp(X>x).
\]

\begin{definition}\label{def:signed-power}
For $p>0$, the \emph{signed power} of $x\in\R$ is
\[
  \spow{x}{p}=\sgn(x)\abs{x}^p,
  \qquad \sgn(0)=0.
\]
The corresponding one-sided moments are
\[
  M_p^+(X)=\E[X_+^p],
  \qquad
  M_p^-(X)=\E[X_-^p],
\]
with values in $[0,\infty]$.  Define the absolute moment, possibly infinite,
by
\[
  A_p(X)=\E[\abs{X}^p]=M_p^+(X)+M_p^-(X).
\]
When $A_p(X)<\infty$, define the signed-power moment by
\[
  S_p(X)=\E[\spow{X}{p}]=M_p^+(X)-M_p^-(X).
\]
\end{definition}

The basic argument is the following layer-cake identity.

\begin{proposition}[Layer-cake identity]\label{prop:layer-cake}
If $Y\geq0$ almost surely and $p>0$, then
\begin{equation}\label{eq:layer-cake}
  \E[Y^p]
  =p\int_0^\infty t^{p-1}\Pp(Y>t)\dd t,
\end{equation}
where both sides are allowed to be infinite.
\end{proposition}

\begin{proof}
For every $y\geq0$,
\[
  y^p=\int_0^\infty p t^{p-1}\ind_{\{t<y\}}\dd t.
\]
The integrand is nonnegative.  Tonelli's theorem therefore permits the
interchange of expectation and integration without first assuming that the
moment is finite:
\[
  \E[Y^p]
  =\int_0^\infty p t^{p-1}\E[\ind_{\{Y>t\}}]\dd t
  =p\int_0^\infty t^{p-1}\Pp(Y>t)\dd t.
\]
\end{proof}

\begin{theorem}[Two-sided tail representations]\label{thm:two-sided}
For every real-valued $X$ and every $p>0$,
\begin{align}
  M_p^+(X)
  &=p\int_0^\infty t^{p-1}\{1-F_X(t)\}\dd t,
  \label{eq:positive-part}\\
  M_p^-(X)
  &=p\int_0^\infty t^{p-1}F_X((-t)^-)\dd t.
  \label{eq:negative-part}
\end{align}
Consequently,
\begin{align}
  A_p(X)
  &=p\int_0^\infty t^{p-1}
    \{1-F_X(t)+F_X((-t)^-)\}\dd t,
  \label{eq:absolute-moment}
\end{align}
as an identity in $[0,\infty]$.  If $A_p(X)<\infty$, then
\begin{align}
  S_p(X)
  &=p\int_0^\infty t^{p-1}
    \{1-F_X(t)-F_X((-t)^-)\}\dd t.
  \label{eq:signed-moment}
\end{align}
\end{theorem}

\begin{proof}
Apply Proposition~\ref{prop:layer-cake} first to $X_+$ and then to $X_-$.
For $t\geq0$,
\[
  \Pp(X_+>t)=\Pp(X>t)=1-F_X(t)
\]
and
\[
  \Pp(X_->t)=\Pp(X<-t)=F_X((-t)^-).
\]
Equations \eqref{eq:absolute-moment} and \eqref{eq:signed-moment} follow
from $\abs{X}^p=X_+^p+X_-^p$ and
$\spow{X}{p}=X_+^p-X_-^p$.
\end{proof}

Related two-sided formulas appear in \citet{NadarajahOkorie2022}.
Corresponding formulas for transforms satisfying the regularity assumptions
of \citet{Liu2020} follow from that paper's general transformed-expectation
framework.  The formulation above separates absolute and signed-power
conventions and uses Tonelli's theorem so that the nonnegative identities
remain valid when a moment is infinite.

If $F_X$ is continuous---in particular, if $X$ is absolutely continuous---
then $F_X((-t)^-)=F_X(-t)$ in these formulas.  No density is otherwise required.

\begin{corollary}[Moment existence]\label{cor:existence}
For $p>0$, $\E[\abs{X}^p]<\infty$ if and only if
\[
  \int_0^\infty t^{p-1}\{1-F_X(t)\}\dd t<\infty
  \quad\text{and}\quad
  \int_0^\infty t^{p-1}F_X((-t)^-)\dd t<\infty.
\]
\end{corollary}

\begin{remark}[Relation to ordinary powers]\label{rem:ordinary-powers}
If $p$ is a positive integer, then $x^p=\abs{x}^p$ for even $p$ and
$x^p=\spow{x}{p}$ for odd $p$.  More generally, suppose
$p=m/n>0$ is in lowest terms and $n$ is odd, and define the power of a
negative number using the real $n$th root.  Then the resulting $x^p$ equals
$\abs{x}^p$ on the negative half-line when $m$ is even and equals
$\spow{x}{p}$ when $m$ is odd.  For a general non-integer real $p$, $x^p$
need not be real when $x<0$.  Absolute and signed powers avoid this domain
problem and vary coherently with every real $p>0$.
\end{remark}

\Needspace{10\baselineskip}
\subsection{Centered representations}

Applying the preceding formulas to $X-a$ yields centered moments and
absolute deviations about any fixed $a\in\R$.

\begin{corollary}[Tail representations around a center]\label{cor:centered}
Let $a\in\R$ and $p>0$.  Then
\begin{equation}
  \E[\abs{X-a}^p]
  =p\int_0^\infty t^{p-1}
    \{1-F_X(a+t)+F_X((a-t)^-)\}\dd t
  \label{eq:centered-absolute}
\end{equation}
as an identity in $[0,\infty]$.  If $\E[\abs{X-a}^p]<\infty$, then
\begin{equation}
  \E[\spow{(X-a)}{p}]
  =p\int_0^\infty t^{p-1}
    \{1-F_X(a+t)-F_X((a-t)^-)\}\dd t.
  \label{eq:centered-signed}
\end{equation}
\end{corollary}

\begin{proof}
Apply Theorem~\ref{thm:two-sided} to $X-a$ and use
\[
  \Pp(X-a>t)=\Pp(X>a+t),
  \qquad
  \Pp(X-a<-t)=\Pp(X<a-t).
\]
\end{proof}

For $p=1$, the signed-power formula is the integrated-tail representation of
\citet[Proposition~2.1]{Lo2019}, written in the present notation; the left
limit records explicitly the strict lower-tail event $\Pp(X<-t)$.
If $\E[\abs X]<\infty$, equations
\eqref{eq:absolute-moment} and \eqref{eq:signed-moment} become
\begin{align}
  \E[\abs X]
  &=\int_0^\infty\{1-F_X(t)+F_X((-t)^-)\}\dd t,
  \label{eq:mean-absolute}\\
  \E[X]
  &=\int_0^\infty\{1-F_X(t)-F_X((-t)^-)\}\dd t.
  \label{eq:mean-signed}
\end{align}
Thus, the mean is the positive-tail area minus the negative-tail area,
whereas the absolute first moment is their sum.

\begin{example}[A continuous two-sided distribution]\label{ex:uniform}
Let $X\sim\operatorname{Unif}(-1,3)$.  Then
\[
  F_X(x)=
  \begin{cases}
    0, & x<-1,\\
    (x+1)/4, & -1\leq x<3,\\
    1, & x\geq3.
  \end{cases}
\]
The two areas in Figure~\ref{fig:uniform} are
\[
  M_1^-(X)=\int_0^1\frac{1-t}{4}\dd t=\frac18,
  \qquad
  M_1^+(X)=\int_0^3\frac{3-t}{4}\dd t=\frac98.
\]
Hence $\E[X]=9/8-1/8=1$ and $\E[\abs X]=9/8+1/8=5/4$.
\end{example}

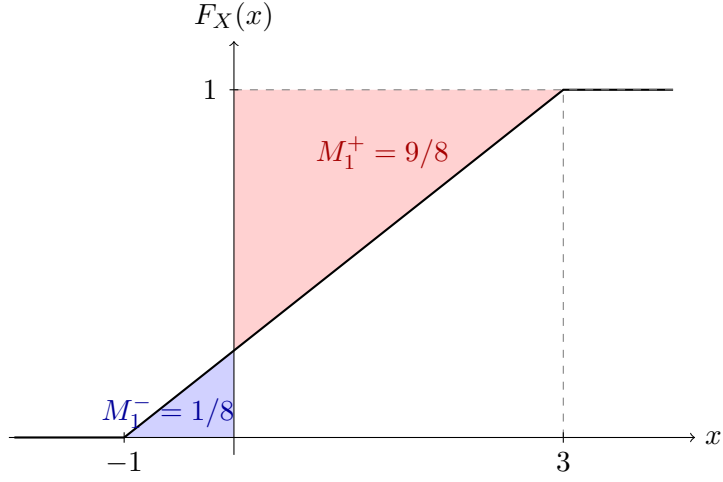
\begin{figure}[H]
\centering
\begin{tikzpicture}[x=1.45cm,y=4.6cm]
  \fill[blue!18] (-1,0)--(0,0.25)--(0,0)--cycle;
  \fill[red!18] (0,0.25)--(3,1)--(0,1)--cycle;
  \draw[->] (-2.05,0)--(4.2,0) node[right] {$x$};
  \draw[->] (0,-0.05)--(0,1.14) node[above] {$F_X(x)$};
  \draw[thick] (-2,0)--(-1,0)--(3,1)--(4,1);
  \draw[dashed,gray] (3,0)--(3,1);
  \draw[dashed,gray] (0,1)--(4,1);
  \foreach \x/\lab in {-1/$-1$,3/$3$}
    \draw (\x,0.015)--(\x,-0.015) node[below] {\lab};
  \draw (-0.04,1)--(0.04,1) node[left=4pt] {$1$};
  \node[blue!60!black] at (-0.6,0.075) {$M_1^-=1/8$};
  \node[red!65!black] at (1.35,0.82) {$M_1^+=9/8$};
\end{tikzpicture}
\caption{For $X\sim\operatorname{Unif}(-1,3)$, the blue area is
$\int_{-1}^{0}F_X(t)\dd t=M_1^-(X)$ and the red area is
$\int_0^3\{1-F_X(t)\}\dd t=M_1^+(X)$.}
\label{fig:uniform}
\end{figure}

\section{Discrete tail sums and summation by parts}
\label{sec:discrete}

For a discrete distribution, a tail integral reduces to a weighted sum of
tail probabilities.  On a lattice, the resulting identity is equivalently a
form of summation by parts.  We first record the ordered-support form used
below.

\begin{proposition}[Ordered nonnegative support]\label{prop:ordered-support}
Let $Y\geq0$ have a possible atom at zero and either finitely many positive
support points $0<a_1<\cdots<a_m$, or countably many positive support points
$0<a_1<a_2<\cdots$ with $a_j\uparrow\infty$.  For $p>0$, in the countably
infinite case,
\begin{equation}\label{eq:general-discrete}
  \E[Y^p]
  =a_1^p\Pp(Y>0)
   +\sum_{j=1}^\infty
    (a_{j+1}^p-a_j^p)\Pp(Y>a_j),
\end{equation}
as an identity in $[0,\infty]$.  In the finite case,
\begin{equation}\label{eq:finite-discrete}
  \E[Y^p]
  =a_1^p\Pp(Y>0)
   +\sum_{j=1}^{m-1}
    (a_{j+1}^p-a_j^p)\Pp(Y>a_j),
\end{equation}
where the sum is empty when $m=1$.
\end{proposition}

\begin{proof}
Partition the integral in \eqref{eq:layer-cake} into $(0,a_1)$ and the
intervals $(a_j,a_{j+1})$.  The tail probability is constant on each open
interval.  Integrating $p t^{p-1}$ gives \eqref{eq:general-discrete}; the
same calculation over finitely many intervals gives
\eqref{eq:finite-discrete}.
\end{proof}

Applying the proposition separately to $X_+$ and $X_-$ gives the
corresponding two-sided formula whenever the positive support points of
$X$ and the magnitudes of its negative support points satisfy these ordered
support assumptions.  This avoids assigning a non-integer real power to a
negative support point.

For a lattice distribution, write
\[
  \Delta_p(k)=(k+1)^p-k^p,
  \qquad k=0,1,2,\ldots.
\]

\begin{theorem}[Integer-valued variables]\label{thm:integer}
If $X$ is integer-valued and $p>0$, then
\begin{align}
  M_p^+(X)
  &=\sum_{k=0}^\infty \Delta_p(k)\{1-F_X(k)\},
  \label{eq:integer-positive}\\
  M_p^-(X)
  &=\sum_{k=1}^\infty \Delta_p(k-1)F_X(-k).
  \label{eq:integer-negative}
\end{align}
Moreover,
\[
  A_p(X)=M_p^+(X)+M_p^-(X)
\]
as an identity in $[0,\infty]$.  If $A_p(X)<\infty$, then
\[
  S_p(X)=M_p^+(X)-M_p^-(X).
\]
In particular, for $p=1$,
\begin{equation}\label{eq:integer-mean}
  \E[X]
  =\sum_{k=0}^\infty\{1-F_X(k)\}
   -\sum_{k=1}^\infty F_X(-k),
\end{equation}
provided $\E[\abs X]<\infty$.
\end{theorem}

\begin{proof}
For an integer-valued variable, the relevant tail probabilities are constant
on every interval $(k,k+1)$, $k=0,1,\ldots$.  Hence
\begin{align*}
  M_p^+(X)
  &=\sum_{k=0}^\infty
    \int_k^{k+1}p t^{p-1}\Pp(X>k)\dd t\\
  &=\sum_{k=0}^\infty\Delta_p(k)\{1-F_X(k)\}.
\end{align*}
Similarly,
\begin{align*}
  M_p^-(X)
  &=\sum_{k=0}^\infty
    \Delta_p(k)\Pp(X<-k)\\
  &=\sum_{k=0}^\infty\Delta_p(k)F_X(-k-1).
\end{align*}
Reindexing the last series gives \eqref{eq:integer-negative}.  The remaining
claims follow from Definition~\ref{def:signed-power}.
\end{proof}

For $X\geq0$, the case $p\geq1$ of \eqref{eq:integer-positive} is the
higher-order tail-sum identity stated by \citet{Chakraborti2019}; the
layer-cake argument above also covers $0<p<1$.

\begin{corollary}[A tail criterion on the nonnegative integers]
\label{cor:lattice-criterion}
Let $Y$ be a nonnegative integer-valued random variable and let $p>0$.
Then
\[
  \E[Y^p]<\infty
  \quad\Longleftrightarrow\quad
  \sum_{k=1}^\infty k^{p-1}\Pp(Y>k)<\infty.
\]
\end{corollary}

\begin{proof}
The mean value theorem, or a direct limit, gives
\[
  \Delta_p(k)\sim p k^{p-1}\qquad(k\to\infty).
\]
The assertion follows from \eqref{eq:integer-positive} and the limit
comparison test.  Finitely many initial terms do not affect convergence.
\end{proof}

\begin{example}[A four-point distribution]\label{ex:four-point}
Suppose
\[
\begin{array}{c|rrrr}
x&-3&-1&2&4\\ \hline
\Pp(X=x)&0.20&0.25&0.30&0.25.
\end{array}
\]
For $p=1$, the nonzero upper-tail terms are
\[
  (\Pp(X>0),\Pp(X>1),\Pp(X>2),\Pp(X>3))
  =(0.55,0.55,0.25,0.25),
\]
so $M_1^+(X)=1.60$.  The nonzero lower-tail terms are
\[
 (F_X(-1),F_X(-2),F_X(-3))=(0.45,0.20,0.20),
\]
so $M_1^-(X)=0.85$.  Consequently,
\[
  \E[X]=1.60-0.85=0.75,
  \qquad
  \E[\abs X]=1.60+0.85=2.45.
\]
The calculation therefore expresses both one-sided moments entirely through
upper- and lower-tail probabilities.
\end{example}

\begin{example}[Zeta distribution]\label{ex:zeta}
Let $Y$ have the zeta distribution
\[
  \Pp(Y=k)=\frac{k^{-\alpha}}{\zeta(\alpha)},
  \qquad k=1,2,\ldots,
  \qquad \alpha>1.
\]
Directly,
\[
  \E[Y^p]
  =\frac{\zeta(\alpha-p)}{\zeta(\alpha)}
\]
when $p<\alpha-1$, and the moment diverges when $p\geq\alpha-1$.
The same threshold follows from Corollary~\ref{cor:lattice-criterion}, since
\[
  \Pp(Y>k)
  =\frac{1}{\zeta(\alpha)}\sum_{j=k+1}^\infty j^{-\alpha}
  \sim\frac{k^{1-\alpha}}{(\alpha-1)\zeta(\alpha)}.
\]
Thus the tail series behaves like a constant multiple of
$\sum k^{p-\alpha}$, which converges exactly when $p<\alpha-1$.
\end{example}

\begin{example}[Skellam distribution]\label{ex:skellam}
Let $N_1\sim\operatorname{Poisson}(\lambda_1)$ and
$N_2\sim\operatorname{Poisson}(\lambda_2)$ be independent, and set
$X=N_1-N_2$.  This is the Skellam distribution
\citep{Skellam1946}.  Since $\E[\abs X]<\infty$,
Theorem~\ref{thm:integer} implies
\[
  \sum_{k=0}^\infty\Pp(X>k)
  -\sum_{k=1}^\infty\Pp(X\leq-k)
  =\lambda_1-\lambda_2.
\]
Table~\ref{tab:skellam} reports the two nonnegative tail sums, evaluated from
the defining Skellam probabilities and rounded to six decimal places.  Their
difference agrees with the known mean.

\begin{table}[h]
\centering
\caption{One-sided first moments for two Skellam distributions.}
\label{tab:skellam}
\begin{tabular}{ccrrr}
\toprule
$\lambda_1$&$\lambda_2$&$M_1^+(X)$&$M_1^-(X)$&$\E[X]$\\
\midrule
3&3&0.956127&0.956127&0\\
1&2&0.267591&1.267591&$-1$\\
\bottomrule
\end{tabular}
\end{table}
\end{example}

\section{Inverse moments and lower-end behavior}
\label{sec:inverse}

The finiteness of positive-order moments is governed by tail decay at
infinity, whereas the finiteness of inverse moments is governed by the
distribution near zero.  Applying the layer-cake identity to $X^{-1}$ makes
this distinction explicit.  Negative moments of positive random variables
are studied, for example, by \citet{ChaoStrawderman1972}.

\begin{theorem}[Inverse moments]\label{thm:inverse}
Let $X>0$ almost surely and let $q>0$.  Then
\begin{equation}\label{eq:inverse}
  \E[X^{-q}]
  =q\int_0^\infty t^{-q-1}F_X(t^-)\dd t
  =q\int_0^\infty t^{-q-1}F_X(t)\dd t,
\end{equation}
as an identity in $[0,\infty]$.
\end{theorem}

\begin{proof}
Proposition~\ref{prop:layer-cake}, applied to $X^{-1}$ with power $q$,
gives
\[
  \E[X^{-q}]
  =q\int_0^\infty s^{q-1}\Pp(X<s^{-1})\dd s
  =q\int_0^\infty s^{q-1}F_X((s^{-1})^-)\dd s.
\]
The substitution $t=s^{-1}$ yields the first integral in
\eqref{eq:inverse}.  A distribution function has at most countably many
jumps.  Hence $F_X(t^-)$ and $F_X(t)$ differ on a Lebesgue-null set, which
gives the second integral.
\end{proof}

The part of \eqref{eq:inverse} over $[1,\infty)$ is always finite because
$F_X(t)\leq1$ and $\int_1^\infty t^{-q-1}\dd t=1/q$.  Existence of the
inverse moment is therefore determined entirely by the behavior of $F_X$
near zero.

\begin{corollary}[Power behavior at zero]\label{cor:near-zero}
Let $X>0$ almost surely and let $q>0$.  Suppose that, as $t\downarrow0$,
\[
  F_X(t)\sim c t^\alpha
\]
for constants $c>0$ and $\alpha>0$.  Then
\[
  \E[X^{-q}]<\infty
  \quad\Longleftrightarrow\quad q<\alpha.
\]
\end{corollary}

\begin{proof}
Near zero, the integrand in \eqref{eq:inverse} is asymptotic to
$c t^{\alpha-q-1}$.  Its integral over $(0,1)$ is finite exactly when
$\alpha-q>0$.
\end{proof}

\begin{example}[Beta distribution]\label{ex:beta}
Let $X\sim\operatorname{Beta}(a,b)$ with $a,b>0$.  Since
\[
  F_X(t)\sim\frac{t^a}{aB(a,b)}
  \qquad(t\downarrow0),
\]
Corollary~\ref{cor:near-zero} shows that $\E[X^{-q}]$ is finite exactly
when $q<a$.  In that case, direct beta integration gives
\[
  \E[X^{-q}]
  =\frac{B(a-q,b)}{B(a,b)}.
\]
The tail criterion agrees with the condition obtained by direct beta
integration and makes explicit the role of the distribution near its lower
endpoint.
\end{example}

\section{Logarithmic moments}
\label{sec:logarithmic}

Let $X>0$ almost surely.  Applying the two-sided mean formula of
\citet[Proposition~2.1]{Lo2019} to $\log X$ and changing variables gives a
distribution-function representation.  We include the direct proof to
display the strict-inequality convention at atoms explicitly.

\begin{corollary}[Distribution-function representation]\label{cor:log-cdf}
Let $X>0$ almost surely.  The positive and negative parts of $\log X$ satisfy
\begin{align}
  \E[(\log X)_+]
  &=\int_1^\infty\frac{1-F_X(t)}{t}\dd t,
  \label{eq:log-positive}\\
  \E[(\log X)_-]
  &=\int_0^1\frac{F_X(t^-)}{t}\dd t,
  \label{eq:log-negative}
\end{align}
as identities in $[0,\infty]$.  Consequently,
$\E[\abs{\log X}]<\infty$ if and only if both integrals are finite; in that
case,
\begin{equation}\label{eq:log-cdf}
  \E[\log X]
  =\int_1^\infty\frac{1-F_X(t)}{t}\dd t
   -\int_0^1\frac{F_X(t^-)}{t}\dd t.
\end{equation}
The left limit in the second integral may be replaced by $F_X(t)$.
\end{corollary}

\begin{proof}
For every $x>0$,
\begin{equation}\label{eq:pointwise-log}
  \log x
  =\int_1^\infty\frac{\ind_{\{x>t\}}}{t}\dd t
   -\int_0^1\frac{\ind_{\{x<t\}}}{t}\dd t.
\end{equation}
The two integrals are $(\log x)_+$ and $(\log x)_-$, respectively.
Tonelli's theorem applied separately to these nonnegative quantities gives
\eqref{eq:log-positive} and \eqref{eq:log-negative}, including when one is
infinite.  Their sum is $\E[\abs{\log X}]$ and, when both are finite, their
difference is $\E[\log X]$.  As in Theorem~\ref{thm:inverse}, replacing a
left limit by the distribution function changes the integrand only on a
Lebesgue-null set.
\end{proof}

The next representation connects the logarithmic moment with the Laplace
transform
\[
  \mathcal L_X(s)=\E[\exp(-sX)],\qquad s\geq0.
\]

\begin{theorem}[Frullani--Laplace representation]\label{thm:log-laplace}
If $X>0$ almost surely and $\E[\abs{\log X}]<\infty$, then
\begin{equation}\label{eq:log-laplace}
  \E[\log X]
  =\int_0^\infty
    \frac{\exp(-s)-\mathcal L_X(s)}{s}\dd s.
\end{equation}
\end{theorem}

\begin{proof}
Frullani's identity \citep{Iyengar1941} gives, for every $x>0$,
\[
  \log x
  =\int_0^\infty\frac{\exp(-s)-\exp(-sx)}{s}\dd s.
\]
For fixed $x$, the numerator has one sign for all $s>0$, and therefore
\[
  \int_0^\infty
  \frac{\abs{\exp(-s)-\exp(-sx)}}{s}\dd s
  =\abs{\log x}.
\]
The assumption $\E[\abs{\log X}]<\infty$ now justifies Fubini's theorem.
Taking expectations inside the integral yields \eqref{eq:log-laplace}.
\end{proof}

The condition in the next corollary is the standard Laplace-transform order
$X\leq_{\mathrm{Lt}}Y$.  The logarithmic comparison is also a classical
consequence of the characterization of this order by functions with
completely monotone derivatives
\citep[Theorem~5.A.4]{ShakedShanthikumar2007}; here it follows directly from
Theorem~\ref{thm:log-laplace}.

\Needspace{8\baselineskip}
\begin{corollary}[Laplace-transform order]\label{cor:laplace-order}
Let $X$ and $Y$ be strictly positive random variables such that
$\E[\abs{\log X}]+\E[\abs{\log Y}]<\infty$.  If
\[
  \mathcal L_X(s)\geq\mathcal L_Y(s)
  \qquad\text{for every }s>0,
\]
then $\E[\log X]\leq\E[\log Y]$.
\end{corollary}

\begin{proof}
Subtract the two absolutely convergent representations in
\eqref{eq:log-laplace}.  The assumed pointwise order makes the resulting
integrand nonpositive.
\end{proof}

\begin{example}[Gamma distributions]\label{ex:gamma-order}
Let $X\sim\operatorname{Gamma}(1,1)$ and
$Y\sim\operatorname{Gamma}(2,2)$, using the shape--rate parameterization.
Both variables have mean one, and
\[
  \mathcal L_X(s)=\frac{1}{1+s},
  \qquad
  \mathcal L_Y(s)=\left(\frac{2}{2+s}\right)^2.
\]
Since
\[
  \mathcal L_X(s)-\mathcal L_Y(s)
  =\frac{s^2}{(1+s)(2+s)^2}\geq0,
\]
Corollary~\ref{cor:laplace-order} gives
$\E[\log X]\leq\E[\log Y]$.  Directly,
\[
  \E[\log X]=-\gamma,
  \qquad
  \E[\log Y]=1-\gamma-\log2,
\]
where $\gamma$ is Euler's constant, confirming the strict inequality.
\end{example}

\begin{remark}[Differentiation at order zero]\label{rem:derivative-zero}
If there is an $\varepsilon>0$ such that
$\E[X^\varepsilon+X^{-\varepsilon}]<\infty$, then for
$\abs r\leq\varepsilon/2$ the mean value theorem and a standard exponential
bound give
\[
  \left\lvert\frac{x^r-1}{r}\right\rvert
  \leq C_\varepsilon\{x^\varepsilon+x^{-\varepsilon}\},
  \qquad x>0,
\]
with the quotient interpreted as $\log x$ at $r=0$.  Dominated convergence
therefore yields
\[
  \left.\frac{\mathrm d}{\mathrm dr}\E[X^r]\right|_{r=0}
  =\E[\log X].
\]
Thus, differentiation at zero is valid under a neighborhood moment
condition.  By contrast, Corollary~\ref{cor:log-cdf} and
Theorem~\ref{thm:log-laplace} require only
$\E[\abs{\log X}]<\infty$.
\end{remark}

\begin{example}[A reciprocal two-point distribution]\label{ex:log-two-point}
Let $\Pp(X=1/2)=\Pp(X=2)=1/2$.  Then
\[
  \int_1^\infty\frac{1-F_X(t)}{t}\dd t
  =\frac12\int_1^2\frac{1}{t}\dd t
  =\frac12\log2,
\]
and
\[
  \int_0^1\frac{F_X(t^-)}{t}\dd t
  =\frac12\int_{1/2}^1\frac{1}{t}\dd t
  =\frac12\log2.
\]
Equation \eqref{eq:log-cdf} therefore gives $\E[\log X]=0$, in agreement
with the direct calculation
$\{\log(1/2)+\log2\}/2=0$.
\end{example}

\section{Summary of domains and finiteness conditions}
\label{sec:summary}

Table~\ref{tab:summary} collects the conditions used in the preceding
sections.  The finiteness of the first two functionals is governed by
behavior in both tails at infinity, that of inverse moments by behavior near
zero, and that of the logarithmic moment by behavior near zero and at
infinity.

\begin{table}[ht]
\centering
\caption{Domains and finiteness conditions for the moment functionals.}
\label{tab:summary}
\small
\begin{tabular}{@{}>{\raggedright\arraybackslash}p{0.18\linewidth}
                  >{\raggedright\arraybackslash}p{0.17\linewidth}
                  >{\raggedright\arraybackslash}p{0.39\linewidth}
                  >{\raggedright\arraybackslash}p{0.17\linewidth}@{}}
\toprule
Functional & Domain & Finiteness condition & Governing region \\
\midrule
$A_p(X)=\E[\abs X^p]$
  & Real-valued $X$; $p>0$
  & Both one-sided integrals in Corollary~\ref{cor:existence} are finite
  & Both tails at infinity \\
$S_p(X)=\E[\spow{X}{p}]$
  & Real-valued $X$; $p>0$
  & $A_p(X)<\infty$ under the convention adopted here
  & Both tails at infinity \\
$\E[X^{-q}]$
  & $X>0$ a.s.; $q>0$
  & $\int_0^1 t^{-q-1}F_X(t)\dd t<\infty$
  & Near zero \\
$\E[\log X]$
  & $X>0$ a.s.
  & Both integrals in \eqref{eq:log-positive}--\eqref{eq:log-negative}
    are finite
  & Near zero and infinity \\
\bottomrule
\end{tabular}
\end{table}

All distribution-function identities are Lebesgue integrals.  Consequently,
$F_X(t)$ and $F_X(t-)$ are interchangeable inside these integrals, although
they differ pointwise at atoms.  Ordinary raw moments are recovered at
integer orders as described in Remark~\ref{rem:ordinary-powers}.

\section{Concluding remarks}

This note organizes tail representations for several moment functionals
through a common set of arguments.  For a real-valued random variable,
separating the positive and negative parts yields absolute and signed-power
moment formulas for every $p>0$, with the treatment of atoms made explicit.
On ordered discrete supports and integer lattices, the same representations
reduce to exact weighted sums of tail probabilities.

For strictly positive random variables, applying the layer-cake identity to
$X^{-1}$ shows that inverse-moment existence is determined by behavior near
zero.  The logarithmic moment admits complementary distribution-function and
Frullani--Laplace representations under the condition
$\E[\abs{\log X}]<\infty$; the Frullani--Laplace representation also gives a
direct comparison under Laplace-transform order.  Taken together, these
results clarify the domains and finiteness conditions associated with
ordinary, absolute, signed-power, inverse, and logarithmic moments across
continuous, discrete, singular, and mixed distributions.

\paragraph{Funding.}
This research was supported in part by the Conselho Nacional de
Desenvolvimento Cient\'ifico e Tecnol\'ogico (CNPq) and the Coordena\c{c}\~ao
de Aperfei\c{c}oamento de Pessoal de N\'ivel Superior (CAPES).

\paragraph{Competing interests.}
The authors declare that they have no competing interests.

\end{document}